\pdfoutput=1
\documentclass{amsproc}
\usepackage{a4wide}
\usepackage{bookmark}
\usepackage{mathtools}

\usepackage{}

\newtheorem{theorem}{Theorem}[section]
\newtheorem{lemma}[theorem]{Lemma}

\theoremstyle{definition}

\theoremstyle{remark}

\numberwithin{equation}{section}
\newcommand\dis{\stackrel{\mathclap{\normalfont\tiny{\mbox{d}}}}{=}}
\begin{document}

 \title[Rate of Convergence of Major Cost]{Rate of Convergence of Major Cost\\ Incurred in the In-Situ Permutation Algorithm}


\author{Sumit Kumar Jha}
\address{Center for Security, Theory, and Algorithmic Research\\ 
International Institute of Information Technology, 
Hyderabad, India}
\curraddr{}
\email{kumarjha.sumit@research.iiit.ac.in}
\thanks{}


\date{}

\begin{abstract}
The in-situ permutation algorithm due to MacLeod replaces $(x_{1},\cdots,x_{n})$ by\\ $(x_{p(1)},\cdots,x_{p(n)})$ where $\pi=(p(1),\cdots,p(n))$ is a permutation of $\{1,2,\cdots,n\}$ using at most $O(1)$ space. Kirshenhofer, Prodinger and Tichy have shown that the major cost incurred in the algorithm satisfies a recurrence similar to sequence of the number of key comparisons needed by the Quicksort algorithm to sort an array of $n$ randomly permuted items. Further, Hwang has proved that the normalized cost converges in distribution. Here, following Neininger and R{\"u}schendorf, we prove the that rate of convergence to be of the order $\Theta(\ln(n)/n)$ in the Zolotarev metric.
\end{abstract}

\maketitle

\section{Introduction}
The in-situ permutation algorithm developed by MacLeod \cite{MacLeod} replaces $(x_{1},\cdots,x_{n})$ by\\ $(x_{p(1)},\cdots,x_{p(n)})$ where $\pi=(p(1),\cdots,p(n))$ is a permutation of $\{1,2,\cdots,n\}$ using at most $O(1)$ space. Kirshenhofer, Prodinger and Tichy \cite{prodinger} have shown that assuming the input comes from a sequence of independently and identically distributed random variables with a common continuous distribution, the major cost measures, say $X_{n}$, incurred in the algorithm, can be described by $X_{0}=0$, and for $n\geq 1$,
\begin{equation}
\label{eq1}
X_{n}\dis X_{I_{n}}+X^{*}_{n-1-I_{n}}+I_{n},
\end{equation} 
where $(X_{n}),(X^{*}_{n})$, $(I_{n})$ are independent, $X_{n}\dis X_{n}^{*}$, and $I_{n}$ is uniformly distributed over $\{0,1,\cdots,n-1\}$. Here the symbol $\dis$ denotes equivalence in distribution.\par 
The mean and variance of $X_{n}$ were calculated by Knuth \cite{knuth} which satisfy 
$$\textbf{E}(X_{n})=n\ln{n}+(\gamma-2)n+O(\ln{n}),\quad \text{Var}(X_{n})=\sigma^{2}n^{2}-n\ln(n)+O(n)$$
where $\gamma$ denotes Euler's constant and $\sigma:=\sqrt{2-\pi^{2}/6}>0$.
\par
Further, Hwang \cite{hwang} showed using R{\"o}sler's contraction method that
$$
Y_{n}:=\frac{X_{n}-\textbf{E}(X_{n})}{n}\xrightarrow{d}Y
$$
where $\xrightarrow{d}$ denotes convergence in distribution. Here $Y$ satisfies 
\begin{equation}
\label{eq2}
Y\dis UY+(1-U)Y^{*}+C(U)
\end{equation}
where $Y\dis Y^{*}$, $U$ is the uniform random variable over the unit interval, $Y,Y^*$, and $U$ are independent, and $C(u):=(1-u)\ln(1-u)+u\ln(u)+u$.\par 
We wish to estimate the rate of convergence $Y_{n}\rightarrow Y$ following Neininger and R{\"u}schendorf \cite{naini}. The basic distance considered in \cite{naini} is the Zolotarev metric $\zeta_{3}$ which given distributions $\mathcal{L}(V),\mathcal{L}(W)$ is defined by
$$\zeta_{3}(\mathcal{L}(V),\mathcal{L}(W)):=\sup_{f\in \mathcal{F}_{3}}|\mathbf{E}f(V)-\mathbf{E}f(W)|,$$
where $\mathcal{F}_{3}:=\{f\in \mathbb{C}^{2}(\mathbb{R},\mathbb{R}):|f''(x)-f''(y)|\leq |x-y|\}$ is the space of all twice differentiable functions with second derivative being Lipschitz continuous with Lipschitz constant $1$. Hereon we use the notation $\zeta_{3}(V,W):=\zeta_{3}(\mathcal{L}(V),\mathcal{L}(W))$. It is known that convergence in $\zeta_{3}$ implies weak convergence and that $\zeta_{3}(V,W)<\infty$ if $\textbf{E}V=\textbf{E}W$, $\textbf{E}V^{2}=\textbf{E}W^{2}$, and $||V||_{3},||W||_{3}<\infty$. The metric $\zeta_{3}$ is ideal of order $3$, that is, we have for $T$ independent of $(V,W)$ and $c\neq 0$
$$\zeta_{3}(V+T,W+T)\leq \zeta_{3}(V,W),\quad \zeta_{3}(cV,cW)=|c|^{3}\zeta_{3}(V,W).$$
We wish to obtain following
\begin{theorem}
\label{thm}
The major cost ($X_{n}$) incurred in the in-situ permutation algorithm satisfying recurrence \eqref{eq1} satisfies
$$\normalfont \zeta_{3}\left(\frac{X_{n}-\textbf{E}(X_{n})}{\sqrt{\text{Var}(X_{n})}},X\right)=\Theta\left(\frac{\log(n)}{n}\right), \quad (n\rightarrow \infty)$$
where $X:=Y/\sigma$ is a scaled version of the limiting distribution in \eqref{eq2}.
\end{theorem}
We modify the proof in \cite{naini} suited for the above case in the next section. \par 
\textsc{Notation:} Subsequently, we use that $\text{Var}(Y)=\sigma$, $||Y||_{3}<\infty$ where $||Y||_{p}:=(\textbf{E}|Y|^{p})^{1/p}, $ $1\leq p<\infty$ denotes the $L^{p}$-norm.
\section{The Proof}
We start with the following lemma from \cite{naini}.
\begin{lemma}
\label{lemma1}
Let $V,W$ have identical first and second moment with $||V||_{3},||W||_{3}<\infty$, then
\begin{equation}
\label{eq3}
\normalfont
\frac{1}{6}|\textbf{E}V^{3}-\textbf{E}W^{3}|\leq \zeta_{3}(V,W)\leq \frac{1}{6}(||V||_{3}^{2}+||V||_{3}\, ||W||_{3}+||W||_{3}^{2})l_{3}(V,W)
\end{equation}
where 
\begin{equation}
\label{eq4}
l_{p}(\mathcal{L}(V),\mathcal{L}(W)):=l_{p}(V,W):=\inf\{||V-W||_{p}:V\dis V,W\dis W\}, \quad p\geq 1.
\end{equation}
\end{lemma}
\begin{proof}[Proof of Theorem \ref{thm}]
The constants $\sigma(n)\geq 0$ are defined by
\begin{equation}
\label{eq5}
\sigma^{2}(n):=\text{Var}(Y_{n})=\sigma^{2}-\frac{\ln(n)}{n}+O\left(\frac{1}{n}\right).
\end{equation}
\textsc{\textbf{Lower Bound:}} Establishing the lower bounds only requires information of moments of $(X_{n})$. Using the bound in lemma \ref{lemma1}, we have

$$\zeta_{3}\left(\frac{X_{n}-\textbf{E}(X_{n})}{\sqrt{\text{Var}(X_{n})}},X\right)\geq \frac{1}{6}
\left|\textbf{E}\left(\frac{Y_{n}}{\sigma(n)}\right)^{3}-\textbf{E}\left(\frac{Y}{\sigma}\right)^{3}\right|$$

Observe that the third moment of $Y_{n}$ is
$$\textbf{E}Y_{n}^{3}=\frac{1}{n^{3}}\textbf{E}(X_{n}-\textbf{E}(X_{n}))^{3}=\frac{1}{n^{3}}\kappa_{3}(X_{n})=M_{3}+O\left(\frac{1}{n}\right)$$
with $M_{3}=\textbf{E}(Y^{3})>0$ where we use the expansion of third cumulant $\kappa_{3}(X_{n})$ of $X_{n}$ which can be explicitly computed using generating functions for factorial moments with aid of Maple in \cite{prodinger}. The equation \eqref{eq5} gives us
$$\frac{1}{\sigma^{3}(n)}=\frac{1}{\sigma^{3}}+\frac{3}{2\sigma^{5}}\frac{\ln(n)}{n}+O\left(\frac{1}{n}\right),$$
and thus
$$\frac{1}{6}
\left|\textbf{E}\left(\frac{Y_{n}}{\sigma(n)}\right)^{3}-\textbf{E}\left(\frac{Y}{\sigma}\right)^{3}\right|=\frac{M_{3}}{4\sigma^{5}}\frac{\ln(n)}{n}+O\left(\frac{1}{n}\right)$$
which proves the claimed lower bound in the theorem.\par 
\textbf{\textsc{Upper Bound}:} The variates $Y_{n}$ would satisfy the recurrence:
\begin{equation}
\label{eq6}
Y_{n}\dis \frac{I_{n}}{n}Y_{I_{n}}+\frac{n-1-I_{n}}{n}Y'_{n-1-I_{n}}+C_{n}(I_{n}), \quad n\geq 1,
\end{equation}
where $(Y_{n}), (Y_{n}'), I_{n}$ are independent, $Y_{k}\dis Y_{k}'$ for all $k\geq 0$ and $C_{n}(k):=\frac{1}{n}(\mu(k)+\mu(n-1-k)-\mu(n)+k)$, 
with 
$\mu(n):=\textbf{E}(X_{n}),$ $n\geq 0$. The rest of the proof follows identically as in the upper bound proof in \cite{naini}.
\end{proof}
\bibliographystyle{amsplain}
\bibliography{sample.bib}
\end{document}